\documentclass{sig-alternate-05-2015}

\usepackage[utf8]{inputenc}

\usepackage{multirow}
\usepackage[pass]{geometry}

\usepackage{graphicx}
\graphicspath{{Pictures/}}
\numberofauthors{1}
\author{
\alignauthor Cl\'ement Pernet
\\
       \affaddr{Univ. Grenoble Alpes}
       \affaddr{Laboratoire LIP}
       \affaddr{{Inria, Université de Lyon}}\\
       \affaddr{46, All\'ee d'Italie, F69364 Lyon Cedex 07, France}\\ 
       \affaddr{\href{mailto:Clement.Pernet@imag.fr}{Clement.Pernet@imag.fr}}
}

\makeatletter
\usepackage{booktabs}
\usepackage{color,svgcolor}
\usepackage[plainpages=true]{hyperref}
\hypersetup{
  pdftitle={Computing with Quasiseparable Matrices},
  pdfauthor={Clément Pernet},
  breaklinks=true, 
  colorlinks=true,
  linkcolor=darkred,
  citecolor=blue,
  urlcolor=darkgreen,
}
\usepackage{algorithm,algpseudocode}
\algrenewcommand\algorithmicrequire{\textbf{Input:}}
\algrenewcommand\algorithmicensure{\textbf{Output:}}
\algrenewcommand\algorithmicreturn{\textbf{Return}}
\algrenewcommand\Return{\State\algorithmicreturn{} }%
\usepackage{ragged2e}

\usepackage{isomath}
\usepackage{amsmath,amssymb}
\usepackage{xspace}
\newtheorem{definition}{Definition}
\newtheorem{corollary}{Corollary}
\newtheorem{theorem}{Theorem}
\newtheorem{remark}{Remark}
\newtheorem{lemma}{Lemma}

\newenvironment{smatrix}{\left[\begin{smallmatrix}}{\end{smallmatrix}\right]}
\newcommand{\mat}[1]{\mathsf{#1}\xspace}

\newcommand{\RPM}{\ensuremath{\mathcal{R}}\xspace}
\newcommand{\LTP}{\ensuremath{\text{Left}}\xspace}
\newcommand{\GO}[1]{\ensuremath{O(#1)}\xspace}
\newcommand{\SO}[1]{\ensuremath{O\tilde\ (#1)}\xspace}
\newcommand{\pluq}{\texttt{PLUQ}\xspace}
\newcommand{\trsm}{\texttt{TRSM}\xspace}
\newcommand{\MM}{\texttt{MM}\xspace}
\newcommand{\bruhat}{Bruhat generator\xspace}

\newlength{\myfigsize}
\title{Computing with Quasiseparable Matrices
}

\makeatother

\begin{document}
\CopyrightYear{2016}
\setcopyright{licensedothergov}
\conferenceinfo{ISSAC '16,}{July 19 - 22, 2016, Waterloo, ON, Canada}
\isbn{978-1-4503-4380-0/16/07}\acmPrice{\$15.00 \\ 
Copyright is held by the owner/author(s). Publication rights licensed to ACM.
}
\doi{http://dx.doi.org/10.1145/2930889.2930915}

\maketitle

\abstract{The class of quasiseparable matrices is defined by a pair of bounds, called
  the quasiseparable orders, on the ranks of the sub-matrices
  entirely located in their strictly lower and upper triangular parts. 
These arise naturally in applications, as e.g. the inverse of band matrices, and are
widely used for they admit structured representations allowing to compute with them in time
linear in the dimension.
We show, in this paper, the connection between the notion of quasiseparability and the rank
profile matrix invariant, presented in [Dumas \& al. ISSAC'15]. This allows us to 
propose an algorithm computing the quasiseparable orders $(r_L,r_U)$ in time
\GO{n^2s^{\omega-2}} where $s=\max(r_L,r_U)$ and $\omega$ the
exponent of matrix multiplication. We then present two new structured representations,
a binary tree of PLUQ decompositions, and the \bruhat, using respectively
$\GO{ns\log\frac{n}{s}}$ and $\GO{ns}$ field elements instead of  $\GO{ns^2}$
for the classical generator and $\GO{ns\log n}$ for the hierarchically semiseparable
representations. We present algorithms computing these
representations in time $\GO{n^2s^{\omega-2}}$. These representations allow a matrix-vector
product in time linear in the size of their representation. Lastly we show
how to multiply two such structured matrices in time $\GO{n^2s^{\omega-2}}$.}






\section{Introduction}

The inverse of a tridiagonal matrix, when it exists, is a dense matrix with the
property that all sub-matrices entirely below or above its diagonal have rank at
most one. This
property and many generalizations of it, defining the semiseparable and
quasiseparable matrices, have been extensively studied over the past 80 years.
We refer to~\cite{VBGM05} and~\cite{VVM07} for a broad bibliographic overview on
the topic.
In this paper, we will focus on the class of quasiseparable matrices, introduced
in~\cite{EiGo99}:


\begin{definition}\label{def:quasisep}
  An $n\times n$ matrix $\mat{M}$ is $(r_L,r_U)$-quasi\-se\-pa\-ra\-ble if its
  strictly lower and upper  triangular parts satisfy the following low rank structure:
  for all $1\leq k\leq n-1$,
  \begin{eqnarray}
    \text{rank}(\mat{M}_{k+1..n,1..k})&\leq& r_L, \label{eq:quasisep:low}\\
    \text{rank}(\mat{M}_{1..k,k+1..n})&\leq& r_U \label{eq:quasisep:up}.
  \end{eqnarray}
The values $r_L$ and $r_U$ 
define the quasiseparable orders of $\mat{M}$.
\end{definition}

Quasiseparable matrices can be represented with fewer than $n^2$ coefficients,
using a structured representation, called a generator. The most commonly used
generator~\cite{EiGo99,VBGM05,VVM07,EGO05,BEG16} for a matrix $\mat{M}$, consists of $(n-1)$ pairs of
vectors $p(i),q(i)$ of size $r_L$, $(n-1)$ pairs of vectors $g(i),h(i)$ of size
$r_U$,  $n-1$ matrices $a(i)$ of dimension $r_L\times r_L$, and  $n-1$
matrices $b(i)$ of dimension $r_U\times r_U$ such that 
$$
\mat{M}_{i,j} = \left\{
  \begin{array}{ll}
    p(i)^T\mat{a}^{>}_{ij} q(j), & 1\leq j<i\leq n\\
    d(i),& 1\leq i=j\leq n\\
    g(i)^T\mat{b}^{<}_{ij} h(j), & 1\leq i<j\leq n\\
  \end{array}
\right.
$$
where $\mat{a}^{>}_{ij} = \mat{a}(i-1)\dots \mat{a}(j+1)$ for $j>i+1$, $\mat{a}_{j+1,j}=1$,
and  $\mat{b}^{<}_{ij} = \mat{b}(i+1)\dots \mat{b}(i-1)$ for $i>j+1$, $b_{i,i+1}=1$.
This representation, of size $\GO{n(r_L^2+r_U^2)}$ makes it possible to apply a
  vector in $\GO{n(r_L^2+r_U^2)}$ field operations, multiply two quasiseparable
  matrices in time $\GO{n\max(r_L,r_U)^3}$ and also compute the inverse in time
  $\GO{n\max(r_L,r_U)^3}$~\cite{EiGo99}. 

The contribution of this paper, is to make the connection between the notion of quasiseparability
and a matrix invariant, the rank profile matrix, that we introduced
in~\cite{DPS15}. More precisely, we show that the 
PLUQ decompositions of the lower and 
upper triangular parts of a quasiseparable matrix, using a certain class of pivoting
strategies, also have a structure ensuring that their memory footprint and the
time complexity to compute them does not depend on the rank 
of the matrix but on the quasiseparable order (which can be arbitrarily lower).
Note that we will assume throughout the paper that the PLUQ decomposition
algorithms mentioned have the ability to reveal ranks. This is the case when
computing with exact arithmetic (e.g. finite fields or multiprecision
rationals), but not always with finite precision floating point
arithmetic. In the latter context, a special care need to be taken for the
pivoting of LU decompositions~\cite{HWY92,Pan00}, and QR or SVD decompositions are often
more commonly used~\cite{Chan87,ChIp94}.
This study is motivated by the design of new algorithms on polynomial matrices where
quasiseparable matrices naturally occur, and  more generally by the framework of the
\texttt{LinBox} library~\cite{linbox:2016} for black-box exact linear algebra. 

After defining and recalling the properties of the rank profile matrix in
Section~\ref{sec:prelim}, we propose in Section~\ref{sec:comporders} an
algorithm computing the quasiseparable orders $(r_L,r_U)$ in time 
\GO{n^2s^{\omega-2}} where $s=\max(r_L,r_U)$ and $\omega$ the
exponent of matrix multiplication. We then present in
Section~\ref{sec:generators} two new structured representations, 
a binary tree of PLUQ decompositions, and the \bruhat, using respectively
$\GO{ns\log\frac{n}{s}}$ and $\GO{ns}$ field elements instead of  $\GO{ns^2}$
for the previously known generators. We present in Section~\ref{sec:cost}
algorithms computing them in time 
$\GO{n^2s^{\omega-2}}$. These representations support a matrix-vector
product in time linear in the size of their representation. Lastly we show
how to multiply two such structured matrices in time
 $\GO{n^2s^{\omega-2}}$.

Throughout the paper, $\mat{A}_{i..j,k..l}$ will denote the sub-matrix of
$\mat{A}$ of row indices between $i$ and $j$ and column indices between $k$ and $l$.
The matrix of the canonical basis, with a one at position $(i,j)$ will be
denoted by~$\mat{\Delta}^{(i,j)}$.
 
\section{Preliminaries}\label{sec:prelim}
\subsection{Left triangular matrices}
We will make intensive use of matrices with non-zero elements only located above the main anti-diagonal.
We will refer to these matrices as left triangular, to avoid any confusion with
upper triangular matrices.
\begin{definition}
  A left triangular matrix is any $m\times n$ matrix $\mat{A}$ such
  that $\mat{A}_{i,j}=0$ for all $i> n-j$.
\end{definition}

The left triangular part of a matrix $\mat{A}$, denoted by $\LTP(\mat{A})$ will
refer to the left triangular matrix extracted from it.
We will need the following property on the left triangular part of the product
of a matrix by a triangular matrix.

\begin{lemma}\label{lem:leftproductup}
  Let  $\mat{A}=\mat{B}\mat{U}$ be an $m\times n$ matrix  where 
$\mat{U}$ is $n\times n$ upper triangular. Then
  $\LTP(\mat{A}) = \LTP(\LTP(\mat{B}) \mat{U})$.
\end{lemma}

\begin{proof}
  Let $\mat{\bar A} = \LTP(\mat{A}), \mat{\bar B} =\LTP(\mat{B})$.
For $j\leq n-i,$ we have
$
\mat{\bar A}_{i,j} = \sum_{k=1}^n \mat{B}_{i,k} \cdot \mat{U}_{k,j} = \sum_{k=1}^j \mat{B}_{i,k} \cdot \mat{U}_{k,j}
$
as $\mat{U}$ is upper triangular. Now for $k\leq j \leq n-i$, $\mat{B}_{i,k} =
\mat{\bar B}_{i,k}$,
which
proves that the left triangular part of $\mat{A}$ is that of $\LTP(\mat{B})\mat{U}$.
\end{proof}

Applying Lemma~\ref{lem:leftproductup} on $\mat{A}^T$ yields Lemma~\ref{lem:leftproductlow}
\begin{lemma}\label{lem:leftproductlow}
  Let $\mat{A}=\mat{L}\mat{B}$  be an $m\times n$ matrix where
$\mat{L}$ is $m\times m$ lower triangular. Then
  $\LTP(\mat{A}) = \LTP(\mat{L}\LTP(\mat{B}))$.
\end{lemma}

Lastly, we will extend the notion of quasiseparable order to left triangular
matrices, in the natural way: the left quasiseparable order is the maximal rank of any leading $k\times
(n-k)$ sub-matrix. When no confusion may occur, we will abuse the definition and
simply call it the quasiseparable order.

\subsection{The rank profile matrix}
We will use a matrix invariant, introduced in~\cite[Theorem~1]{DPS15}, that summarizes the
information on the ranks of any leading sub-matrices of a given input matrix.

\begin{definition}{\cite[Theorem~1]{DPS15}}\label{def:rpm}
  The rank profile matrix of an $m\times n$ matrix $\mat{A}$ of rank $r$ is the unique $m \times
  n$ matrix $\mathcal{R}_\mat{A}$, with only $r$ non-zero coefficients, all equal to one, located on
  distinct rows and columns such that any leading sub-matrices of $\mathcal{R}_\mat{A}$ has
  the same rank as the corresponding leading sub-matrix in $\mat{A}$.
\end{definition}

This invariant can be computed in just one Gaussian elimination of the matrix
$\mat{A}$, at the cost of $\GO{mnr^{\omega-2}}$ field operations~\cite{DPS15}, provided some conditions on the pivoting strategy being used. It is obtained
from the corresponding PLUQ decomposition as the product 
$$\RPM_A = \mat{P}\begin{bmatrix}  \mat{I}_r\\&\mat{0}_{(m-r)\times (n-r)}\end{bmatrix}\mat{Q}.$$
%
We also recall in Theorem~\ref{th:PLPT} an important property of such PLUQ
decompositions revealing the rank profile matrix.

\begin{theorem} [{\cite[Th.~24]{DPS15:JSC}, \cite[Th.~1]{DPS13}}] \label{th:PLPT}
Let $\mat{A}=\mat{P}\mat{L}\mat{U}\mat{Q}$ be a PLUQ decomposition revealing the
rank profile matrix of $\mat{A}$.
Then, 
$\mat{P}\begin{bmatrix}  \mat{L}&\mat{0}_{m\times (m-r)}\end{bmatrix}\mat{P}^T$ 
is lower triangular and 
$\mat{Q}^T\begin{bmatrix}  \mat{U}\\\mat{0}_{(n-r)\times n}\end{bmatrix}\mat{Q}$ 
is upper triangular.
\end{theorem}

\begin{lemma}\label{lem:preserverpm} The rank profile matrix invariant is preserved by multiplication
  \begin{enumerate}
  \item to the left  with an invertible lower triangular matrix,
  \item to the right with an invertible upper triangular matrix.
  \end{enumerate}
\end{lemma}
\begin{proof}
  Let $B=LA$ for an invertible lower triangular matrix $L$. Then $\text{rank}(B_{1..i,1..j})
  = \text{rank}(L_{1..i,1..i} A_{1..i,1..j}) = \text{rank}(A_{1..i,1..j})$ for
  any $i\leq m, j\leq n$. Hence $\RPM_B = \RPM_A$.
\end{proof}
\section{Computing the quasiseparable orders}\label{sec:comporders}

Let $\mat{M}$ be an $n\times n$ matrix of which one want to determine the
quasiseparable orders $(r_L,r_U)$. Let $\mat{L}$ and $\mat{U}$ be respectively the
lower triangular part and the upper triangular part of $\mat{M}$. 

Let $\mat{J}_n$ be the unit anti-diagonal matrix. Multiplying on the left by
$\mat{J}_n$ reverts the row order while multiplying on the right by $\mat{J}_n$ reverts
the column order.
Hence both $\mat{J}_n \mat{L}$ and
$\mat{U}\mat{J}_n$ are left triangular matrices. Remark that the
conditions~\eqref{eq:quasisep:low} and~\eqref{eq:quasisep:up} state that all 
leading $k\times (n-k)$ sub-matrices of $\mat{J}_n \mat{L}$ and
$\mat{U}\mat{J}_n$ have rank no greater than $r_L$ and $r_U$ respectively.
We will then use the rank profile matrix of these two left triangular matrices
to find these parameters.

\subsection{From a rank profile matrix}

First, note that the rank profile matrix of a left triangular matrix is not necessarily
left triangular. For example, the rank profile matrix of $
\begin{smatrix}
  1 & 1 & 0\\
  1 & 0 & 0\\
  0 & 0 & 0
\end{smatrix}
$
is
$
\begin{smatrix}
  1 &0&0\\
  0 &1&0\\
  0 &0 &0
\end{smatrix}
$.
However, only the left triangular part of the rank profile matrix is sufficient to
compute the left quasiseparable orders. 

Suppose for the moment that the left-triangular part of the rank profile matrix
of a left triangular matrix is given (returned by a function LT-RPM). It remains to enumerate all leading $k\times (n-k)$
sub-matrices and find the one with the largest number of non-zero elements.
Algorithm~\ref{alg:qsrank} shows how to compute the largest rank of all leading
sub-matrices of such a matrix. Run on $\mat{J}_n\mat{L}$ and $\mat{U}\mat{J}_n$, it returns
successively the quasiseparable orders $r_L$ and $r_U$.
\begin{algorithm}[h]
{\small
  \caption{QS-order} \label{alg:qsrank}
  \begin{algorithmic}
    \Require{$\mat{A}$, an $n\times n$ matrix}
    \Ensure{$\max\{\text{rank}(\mat{A}_{1..k,1..n-k}) : 1\leq k\leq n-1\}$}
    \State $\mat{R} \leftarrow \text{LT-RPM}(\mat{A})$\Comment{The left
      triangular part of the rank profile matrix of $\mat{A}$}
      \State rows $\leftarrow $ (False,\dots,False)
      \State cols $\leftarrow $ (False,\dots,False)
      \ForAll{$(i,j)$ such that $\mat{R}_{i,j}=1$}
        \State rows[i] $\leftarrow$ True
        \State cols[j] $\leftarrow$ True
      \EndFor
      \State $s,r \leftarrow 0$
      \For{$i=1\dots n$}
        \State \textbf{if} rows[$i$] \textbf{then} $r\leftarrow r+1$
        \State \textbf{if} cols[$n-i+1$] \textbf{then} $r\leftarrow r-1$
        \State $s\leftarrow \max(s,r)$
      \EndFor
      \State \textbf{return} $s$
  \end{algorithmic}
}
\end{algorithm}
  
This algorithm runs in $\GO{n}$ provided that the rank profile matrix $\mat{R}$
is stored in a compact way, e.g. using a vector of $r$ pairs of pivot indices ($[(i_1,j_1),\dots,(i_r,j_r)]$.

\subsection{Computing the rank profile matrix of a left triangular matrix}

We now deal with the missing component: computing the left triangular part of
the rank profile matrix of a left triangular matrix.

\subsubsection{From a PLUQ decomposition}

A first approach is to run any Gaussian elimination algorithm that can reveal
the rank profile matrix, as described in~\cite{DPS15}. In
particular, the PLUQ decomposition algorithm of~\cite{DPS13} computes the rank
profile matrix of $\mat{A}$ in $\GO{n^2r^{\omega-2}}$ where
$r=\text{rank}(\mat{A})$. 
However this estimate is pessimistic as it does not take into account the left
triangular shape of the matrix. 
Moreover, this estimate does not depend on the left quasiseparable order $s$
but on the rank $r$, which may be much higher.

\begin{remark}\label{rem:tradeoff}
The discrepancy between the rank $r$ of a left triangular matrix and its
quasiseparable order arises from the location of the pivots in its rank profile
matrix.
Pivots located near the top left corner of the matrix  are shared by
many leading sub-matrices, and are therefore likely contribute to the
quasiseparable order. On the other hand, pivots near the anti-diagonal can
be numerous, but do not add up to a large quasiseparable order.
As an illustration, consider the two following extreme cases:
\begin{enumerate}
\item a matrix $\mat{A}$ with generic rank profile. Then the leading
  $r\times r$ sub-matrix of $\mat{A}$ has rank $r$ and the 
  quasiseparable order is $s=r$.
\item  the matrix with $n-1$ ones right above the anti-diagonal. It has rank
  $r=n-1$ but quasiseparable order $1$.
\end{enumerate}
\end{remark}

Remark~\ref{rem:tradeoff} indicates that in the unlucky cases when $r\gg s$,
the computation should reduce to instances of smaller sizes, hence a trade-off
should exist between, on one hand, the discrepency between $r$ and $s$, and
on the other hand, the dimension $n$ of the problems. All contributions
presented in the remaining of the paper are based on such trade-offs.

\subsubsection{A dedicated algorithm}

In order to reach a complexity depending on $s$ and not $r$, we adapt in
Algorithm~\ref{alg:LTElim}  the tile recursive 
algorithm of~\cite{DPS13}, so that the left triangular structure of the input
matrix is preserved and can be used to reduce the amount of computation.

Algorithm~\ref{alg:LTElim} does not assume that the input matrix is
left triangular, as it will be called recursively with arbitrary matrices, but
guarantees to return the left triangular part of the rank profile matrix. 
\begin{algorithm}[h]
{\small
  \begin{algorithmic}[1]
    \caption{LT-RPM: Left Triangular part of the Rank Profile Matrix} \label{alg:LTElim}
    \Require{$\mat{A}$: an $n\times n$ matrix}
    \Ensure{$\mathcal{R}$: the left triangular part of the RPM of $\mat{A}$}
    \State \textbf{if} $n=1$ \textbf{then} \textbf{return} $[0]$
    \State {Split $\mat{A} = \begin{bmatrix} \mat{A_{1}} & \mat{A_{2}}\\\mat{A_{3}} \end{bmatrix}$ 
      where $\mat{A_{3}}$ is $\lfloor \frac{n}{2} \rfloor \times \lfloor
      \frac{n}{2}\rfloor$}
    \State Decompose $\mat{A_{1}} =
    \mat{P_1} \begin{bmatrix}\mat{L_1}\\\mat{M_1}\end{bmatrix}\begin{bmatrix}\mat{U_1}&\mat{V_1}\end{bmatrix} \mat{Q_1}$ 
    \State $\mathcal{R}_1  \leftarrow \mat{P_1} \begin{bmatrix}
      \mat{I_{r_1}}\\&\mat{0}\end{bmatrix}\mat{Q_1}$ where $r_1=\text{rank}(\mat{A_1})$.
    \State $\begin{bmatrix} \mat{B_1}\\ \mat{B_2}\end{bmatrix}
    \leftarrow \mat{P_1}^T\mat{A_{2}}$ 
    \State $
    \begin{bmatrix}
      \mat{C_1}&\mat{C_2}
    \end{bmatrix}
    \leftarrow \mat{A_{3}}\mat{Q_1}^T$ 
    \State Here $A =
    \left[\begin{array}{cc|c}
        \mat{L_1} \backslash \mat{U_1}& \mat{V_1}& \mat{B_1}\\
        \mat{M_1}               & \mat{0}  & \mat{B_2}\\
        \hline
        \mat{C_1}               & \mat{C_2}& \\
      \end{array}\right]$.
    \State $\mat{D}\leftarrow \mat{L_1}^{-1}\mat{B_1}$ 
    \State $\mat{E}\leftarrow \mat{C_1}\mat{U_1}^{-1}$ 
    \State $\mat{F}\leftarrow \mat{B_2}-\mat{M_1}\mat{D}$ 
    \State $\mat{G}\leftarrow \mat{C_2}-\mat{E}\mat{V_1}$
    \State Here $\mat{A}=
    \left[\begin{array}{cc|c}
        \mat{L_1} \backslash \mat{U_1}& \mat{V_1}& \mat{D}\\
        \mat{M_1}               & \mat{0}  & \mat{F}\\
        \hline
        \mat{E}               & \mat{G}& \\
      \end{array}\right]$.
    
    \State $\mat{H} \leftarrow \mat{P_1}  \begin{bmatrix} \mat{0}_{r_1 \times \frac{n}{2}} \\ \mat{F} \end{bmatrix}$
    \State $\mat{I} \leftarrow  \begin{bmatrix} \mat{0}_{r_1 \times \frac{n}{2}} & \mat{G} \end{bmatrix} \mat{Q_1}$
    \State $\mathcal{R}_2 \leftarrow \texttt{LT-RPM}(\mat{H})$
    \State $\mathcal{R}_3 \leftarrow \texttt{LT-RPM}(\mat{I})$
    \State \textbf{return} $\mathcal{R} \leftarrow  \begin{bmatrix}     \mathcal{R}_1 & \mathcal{R}_2\\ \mathcal{R}_3   \end{bmatrix}$
  \end{algorithmic}
}
\end{algorithm}
While the top left quadrant $\mat{A}_1$ is eliminated using any PLUQ decomposition algorithm
revealing the rank profile matrix, the top right and bottom left quadrants
are handled recursively.

\begin{theorem}\label{th:LTRPM}
  Given an $n\times n$ input matrix $\mat{A}$ with left quasiseparable order $s$,
  Algorithm~\ref{alg:LTElim} computes the left triangular part of the rank 
  profile matrix of $\mat{A}$ in $\GO{n^2s^{\omega-2}}$.
\end{theorem}

\begin{proof}
First remark that 
$$
\mat{P_1} \begin{bmatrix}  \mat{D}\\\mat{F}\end{bmatrix} = \underbrace{\mat{P_1}
\begin{bmatrix} \mat{L_1}^{-1}\\ -\mat{M_1}\mat{L_1}^{-1} & \mat{I_{n-r_1}} \end{bmatrix} \mat{P_1}^T}_{\mat{L}} \mat{P_1}
\begin{bmatrix}   \mat{B_1}\\\mat{B_2} \end{bmatrix} = \mat{L}\mat{A}_2.
$$
Hence
$$
\mat{L} \begin{bmatrix}  \mat{A_1} & \mat{A_2}\end{bmatrix} =
\mat{P_1}
\left[\begin{array}{c|c}
  \begin{bmatrix}\mat{U_1} & \mat{V_1}\end{bmatrix}\mat{Q_1} & \mat{D} \\
   \mat{0}   & \mat{F}
\end{array}
\right].
$$
From Theorem~\ref{th:PLPT}, the matrix $\mat{L}$ is lower triangular and by Lemma~\ref{lem:preserverpm}
the rank profile matrix of $\begin{bmatrix}  \mat{A_1} & \mat{A_2}\end{bmatrix}$ equals  that of $\mat{P_1}
\left[\begin{array}{c|c}
  \begin{bmatrix}\mat{U_1} & \mat{V_1}\end{bmatrix}\mat{Q_1} & \mat{D} \\
   \mat{0}   & \mat{F}
\end{array}
\right]$.
Now as $\mat{U_1}$ is upper triangular and non-singular, this rank profile matrix is
in turn that of
$\mat{P_1}
\left[\begin{array}{c|c}
  \begin{bmatrix}\mat{U_1} & \mat{V_1}\end{bmatrix}\mat{Q_1} & \mat{0} \\
   \mat{0}   & \mat{F}
\end{array}
\right]$ and its left triangular part is $\begin{bmatrix}  \mathcal{R}_1 & \mathcal{R}_2\end{bmatrix}$.

By a similar reasoning, $\begin{bmatrix}  \mathcal{R}_1& \mathcal{R}_3\end{bmatrix}^T$
is the left triangular part of the rank profile matrix of $\begin{bmatrix}  \mat{A_1}& \mat{A_3}\end{bmatrix}^T$, which shows that the algorithm is correct.

Let $s_1$ be the left quasiseparable order of $\mat{H}$ and $s_2$ that of $\mat{I}$.
The number of field operations to run Algorithm~\ref{alg:LTElim} is
$$
T(n,s)=\alpha
r_1^{\omega-2}n^2+T_{\text{LT-RPM}}(n/2,s_1)+T_{\text{LT-RPM}}(n/2,s_2)
$$
 for a
positive constant $\alpha$.
We will prove by induction that $T(n,s)\leq 2\alpha s^{\omega-2}n^2$.

Again, since $\mat{L}$ is lower triangular, the rank profile matrix of
$\mat{L}\mat{A_2}$ is that of $\mat{A}_2$ and the quasiseparable orders of
the two matrices are the same. Now $\mat{H}$ is the matrix $\mat{L}\mat{A_2}$
with some rows zeroed out, hence $s_1$, the
quasiseparable order of $\mat{H}$ is no greater than that of $\mat{A_2}$ which 
is less or equal to $s$.
Hence $\max(r_1,s_1,s_2) \leq s $ and we obtain $T(n,s) \leq  \alpha
s^{\omega-2}n^2 + 4 \alpha s^{\omega-2}(n/2)^2 = 2\alpha s^{\omega-2}n^2$.
\end{proof}


\section{More compact generators} \label{sec:generators}

Taking advantage of their low rank property, quasiseparable matrices can be
represented by a structured representation allowing to compute efficiently with them, as for example in
the context of QR or QZ elimination~\cite{EGO05,BEG16}.

The most commonly used generator, as described in~\cite{EiGo99,BEG16} and in the introduction, 
represents an $(r_L,r_U)$-quasiseparable matrix of order $n$  by
$\GO{n(r_L^2+r_U^2)}$ field coefficients\footnote{Note that the statement of 
$\GO{n(r_L+r_U)}$ for the same generator in
~\cite{EGO05} is
erroneous.
}. 

Alternatively, hierarchically semiseparable representations (HSS)
\cite{XCSGL10,KHC16} use numerical rank revealing factorizations of the
off-diagonal blocks in a divide and conquer approach, reducing the size to
\GO{\max(r_L,r_U)n\log n}\cite{KHC16}.
 
A third approach, based on Givens or unitary weights~\cite{DVB07}, performs
another kind of elimination so as to compact the low rank off-diagonal blocks
of the input matrix.

We propose, in this section, two alternative generators, based on an exact PLUQ
decomposition revealing the rank profile matrix. The first one matches the  best
space complexity of the  HSS representation, and improves the time complexity to
compute it by a reduction to fast matrix multiplication. The second one also improves on the space
complexity of HSS representation by removing the extra $\log n$ factor and
shares some similarities with the unitary weight representations of~\cite{DVB07}.

First, remark that storing a PLUQ decomposition 
of rank $r$ and dimension $n\times n$ uses $2rn-r^2$
coefficients: each of the $\mat{L}$ and $\mat{U}$ factor has dimension 
$n\times r$ or $r\times n$; the negative $r^2$ term comes from the lower and upper
triangular shapes of $\mat{L}$ and $\mat{U}$. 
Here again, the rank $r$ can be larger than the quasiseparable order $s$ thus
storing directly a PLUQ decomposition is too expensive.
But as in Remark~\ref{rem:tradeoff}, the setting where $r\gg s$
is precisely when the pivots are near the anti-diagonal, and therefore the $L$
and $U$ factors have an additional structure, with numerous zeros.
%
The two proposed generators, rely on this  fact.

\subsection{A binary tree of PLUQ decompositions}

Following the divide and conquer scheme of Algorithm~\ref{alg:LTElim}, we
propose a first generator requiring
\begin{equation}\label{eq:mem}
\GO{n(r_L\log \frac{n}{r_L}+r_U\log  \frac{n}{r_U})}
\end{equation}
 coefficients.

For a left triangular matrix $\mat{A}=\begin{bmatrix}  \mat{A_1} &
  \mat{A_2}\\ \mat{A_3}\end{bmatrix}$, the sub-matrix $\mat{A_1}$ is represented
by its PLUQ decomposition $(\mat{P}_1,\mat{L}_1,\mat{U}_1,\mat{Q}_1)$, which
requires $2r_1 \frac{n}{2} \leq sn$ field coefficients for 
$\mat{L}_1$ and $\mat{U}_1$ and $2n$ indices for $P$ and $Q$. This scheme is
then recursively applied for the representation of $\mat{A_2}$ and $\mat{A_3}$. 
These matrices have quasiseparable order at most $s$, therefore the following
recurrence relation  for the size of the representation holds:
$$
\left\{
  \begin{array}{llll}
    S(n,s) &=&  sn + 2S(n/2,s) & \text{for } s< n/2\\
    S(n,s) &=&  \frac{n^2}{2} + 2S(n/2,n/4) & \text{for } s\geq n/2\\
\end{array}\right.$$
For $s\geq n/2$, it solves in $S(n,s)=n^2$. Then for $s<n/2$, 
$S(n,s)=sn+2sn/2+\dots+2^kS(n/2^k,s)$, for $k$ such that $\frac{n}{2^k}\leq
s<\frac{n}{2^{k-1}}$, which is $k=\lceil\log_2\frac{n}{s}\rceil$. Hence
 $S(n,s)=sn\log_2 \frac{n}{s} + sn = \GO{sn\log\frac{n}{s}}$.
The estimate~\eqref{eq:mem} is obtained by applying this generator to the upper
and lower triangular parts of the $(r_L,r_U)$-quasiseparable matrix.

This first generator does not take fully advantage of the rank structure of the matrix:
the representation of each anti-diagonal block is independent from the pivots found in
the block $\mat{A}_1$. The second generator, that will be presented in the next
section adresses this issue, in order to remove the logarithmic factors in the estimate~\eqref{eq:mem}.

\subsection{The \bruhat}

We propose an alternative generator inspired by the generalized Bruhat
decomposition~\cite{MH07,Mal10,DPS15:JSC}. Contrarily to the former one, it is not
depending on a specific recursive cutting of the matrix.

Given a left triangular matrix $\mat{A}$ of quasiseparable order $s$ and a PLUQ decomposition of it,
revealing its rank profile matrix $\mat{E}$, the generator consists in the three matrices 
\begin{eqnarray}
  \label{eq:storage}
 \mathcal{L} &=&\LTP(\mat{P}\begin{bmatrix}\mat{L}&\mat{0}\end{bmatrix}\mat{Q}) \label{eq:storage:L},\\
 \mathcal{E} &=& \LTP(\mat{E}),\\
 \mathcal{U} &=& \LTP(\mat{P}\begin{bmatrix}    \mat{U}\\ \mat{0}  \end{bmatrix} \mat{Q})\label{eq:storage:U}.
\end{eqnarray}

Lemma~\ref{lem:LEUstorage} shows that these three matrices suffice to recover
the initial left triangular matrix.
\begin{lemma}\label{lem:LEUstorage}
$\mat{A}= \LTP( \mathcal{L}\mathcal{E}^T \mathcal{U})$
\end{lemma}

\begin{proof}
 $   \mat{A} =  \mat{P}\begin{bmatrix}  \mat{L} & \mat{0}_{m\times (n-r)}\end{bmatrix} \mat{Q}  \mat{Q}^T \begin{bmatrix}  \mat{U}\\
  \mat{0}_{(n-r)\times n}\end{bmatrix} \mat{Q}.
$
From Theorem~\ref{th:PLPT}, the matrix $\mat{Q}^T\begin{bmatrix}  \mat{U}\\  \mat{0}\end{bmatrix} \mat{Q}$
is upper triangular and the matrix $\mat{P}
\begin{bmatrix}  \mat{L}&\mat{0}\end{bmatrix}\mat{P}^T$ is lower triangular.
Applying Lemma~\ref{lem:leftproductup} yields
$
\mat{A}=\LTP(\mat{A}) = \LTP(\mathcal{L}\mat{Q}^T\begin{bmatrix}  \mat{U}\\
  \mat{0}\end{bmatrix} \mat{Q} ) = \LTP(\mathcal{L}\mat{E}^T \mat{P}\begin{bmatrix}  \mat{U}\\  \mat{0}\end{bmatrix} \mat{Q} ),
$
where $\mat{E}=\mat{P}\begin{smatrix}  \mat{I}_r\\&0\end{smatrix}\mat{Q}$.
Then, as $\mathcal{L}\mat{E}^T$ is the matrix $\mat{P}
\begin{bmatrix}  \mat{L}&\mat{0}\end{bmatrix}\mat{P}^T$ with some coefficients
zeroed out, it is lower triangular, hence applying again
Lemma~\ref{lem:leftproductlow} yields
\begin{equation}\label{eq:ALEU}
\mat{A}= \LTP(\mathcal{L}\mat{E}^T \mathcal{U}).
\end{equation}
Consider any non-zero coefficient $e_{j,i}$ of $\mat{E}^T$ that is not in its the left
triangular part, i.e. $j>n-i$. Its contribution to the product
$\mathcal{L}\mat{E}^T$, is only of the form $\mathcal{L}_{k,j}e_{j,i}$. However
the leading coefficient in column $j$ of $\mat{P}\begin{bmatrix}
  \mat{L}&0\end{bmatrix}\mat{Q}$ is precisely at position $(i,j)$. Since
$i>n-j$, this means that the $j$-th column of $\mathcal{L}$ is all zero, and
therefore $e_{i,j}$ has no contribution to the product.
Hence we finally have
 $\mat{A} = \LTP(\mathcal{L}\mathcal{E}^T\mathcal{U})$.
\end{proof}


We now analyze the space required by this generator.
\begin{lemma}\label{lem:size}
  Consider an $n\times n$ left triangular rank profile matrix $\mat{R}$ with quasiseparable order $s$. Then a
  left triangular matrix $\mat{L}$ all zero except at the positions of the pivots of
  $\mat{R}$ and below these pivots, does not contain more than $s(n-s)$ non-zero
  coefficients. 
\end{lemma}

\begin{proof}
  Let $p(k)=\text{rank}(\mat{R}_{1..k,1..n-k})$. The value $p(k)$ indicates the
  number of non zero columns located in the $k\times n-k$ leading sub-matrix of
  $\mat{L}$. Consequently the sum $\sum_{k=1}^{n-1} p(k)$ is an upper bound on the
  number of non-zero coefficients in $\mat{L}$.
  Since $p(k)\leq s$, it is bounded by $sn$. More precisely, there is no more
  than $k$ pivots in the first $k$ columns and the first $k$ rows, hence
  $p(k)\leq k$  and $p(n-k)\leq k$ for $k\leq s$. The bound becomes $s(s+1)
  +(n-2s-1)s = s(n-s)$.
\end{proof}

\begin{corollary}
The generator $\left(\mathcal{L},\mathcal{E},\mathcal{U}\right)$ uses $2s(n-s)$ field
  coefficients and $\GO{n}$ additional indices.
\end{corollary}

\begin{proof}
The leading column elements of  $\mathcal{L}$ are located at
the pivot positions of the left triangular rank profile matrix $\mathcal{E}$.
Lemma~\ref{lem:size} can therefore be applied to show that this matrix occupies
no more than $s(n-s)$ non-zero coefficients.
The same argument applies to the matrix $\mathcal{U}$.
\end{proof}

Figure~\ref{fig:bruhatstorage} illustrates this generator on  a left triangular
matrix of quasiseparable order $5$. 
\begin{figure}[h]
\begin{center}
   \includegraphics[height=\myfigsize]{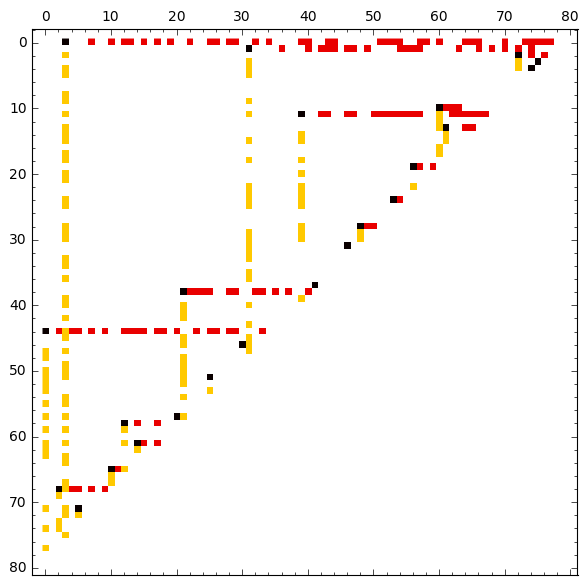}\\
\end{center}
\caption{Support of the $\mathcal{L}$ (yellow), $\mathcal{E}$ (black) and
  $\mathcal{U}$ (red) matrices of the \bruhat for a $80\times 80$ left triangular
  matrix of quasiseparable order $5$.} \label{fig:bruhatstorage}
\end{figure}
As the supports of $\mathcal{L}$ and
$\mathcal{U}$ are disjoint, the two matrices can be shown on the same
left triangular matrix. The pivots of $\mathcal{E}$ (black) are the leading
coefficients of every non-zero row of $\mathcal{U}$ and non-zero column of $\mathcal{L}$.

\begin{corollary}
  Any $(r_L,r_U)$-quasiseparable matrix of dimension $n \times n$ can be
  represented by a generator using no more than $2n(r_L+r_U)+n-2(r_L^2
  -2r_U^2)$ field elements.
\end{corollary}

\subsection{The compact \bruhat}

The sparse structure of the \bruhat  makes it not amenable to the use
of fast matrix arithmetic. 
We therefore propose here a slight variation of it, that we will use in
section~\ref{sec:cost} for fast complexity estimates. We will first describe this compact
representation for the $\mathcal{L}$ factor of the \bruhat.

First, remark that there exists a permutation matrix $\mathcal{Q}$ moving the
non-zero columns of $\mathcal{L}$ to the first $r$  positions,
sorted by increasing leading row index, i.e. such that $\mathcal{L}\mathcal{Q}$
is in column echelon form.
The matrix $\mathcal{L}\mathcal{Q}$ is now compacted, but still has $r=\text{rank}(A)$
columns, which may exceed $s$ and thus preventing to reach
complexities in terms of $n$ and $s$ only. We will again use the
argument of Lemma~\ref{lem:size} to produce a more compact representation with
only $\GO{ns}$ non-zero elements, stored in dense blocks.
Algorithm~\ref{alg:compactbruhat} shows how to build such a representation composed of
a block diagonal matrix and a block sub-diagonal matrix, where all blocks  have
column dimension $s$:
$
\begin{smatrix}
  \mat{D}_{1}\\
  \mat{S}_{2} & \mat{D}_{2}\\
  &\mat{S}_{3} & \mat{D}_{3}\\
  &&\ddots&\ddots\\
  &&&\mat{S}_{t}&\mat{D}_{t}
\end{smatrix}.
$

\begin{algorithm}[h]
{\small
\caption{Compressing the \bruhat} \label{alg:compactbruhat}
 \begin{algorithmic}[1]
\Require{$\mathcal{L}$: the first matrix of the \bruhat}
\Ensure{$\mat{D}, \mat{S}, \mat{T}, \mathcal{Q}$: the compression of $\mathcal{L}$}

\State $\mathcal{Q}\leftarrow$ a permutation s.t. $\mathcal{L}\mathcal{Q}$
is in column echelon form
\State $\mat{C}\leftarrow \mathcal{L}  \mathcal{Q} \begin{smatrix}  \mat{I}_r\\
  \mat{0}\end{smatrix}$ where $r=\text{rank}(\mathcal{L})$ 
\State Split $\mat{C}$ in column slices of width $s$. \\ \Comment $\mat{C}=
\begin{smatrix}
  \mat{C}_{11}& \\
  \mat{C}_{21} & \mat{C}_{22} \\
  \vdots & \vdots& \ddots\\
  \mat{C}_{t1} & \mat{C}_{t2} & \dots & \mat{C}_{tt}
\end{smatrix}$ where $\mat{C}_{ii}$ is $k_i\times s$.
\State $\mat{D} \leftarrow \text{Diag}(\mat{C}_{11},\dots,\mat{C}_{tt})$
\State \label{step:CminusD} $\mat{C} \leftarrow \mat{C}-\mat{D} = 
\begin{smatrix}\mat{0}\\ \mat{C}_{21}  \\
\vdots &  \ddots&\ddots\\
  \mat{C}_{t1} &  \dots & \mat{C}_{t,t-1}& \mat{0}
\end{smatrix}
$
\State $\mat{T}\leftarrow \mat{I}_n$
\For{$i=3\dots t$} \label{step:loop}
   \For{each non zero column $j$ of $ \begin{smatrix}  \mat{C}_{i,i-2} \\ \dots \\
       \mat{C}_{t,i-2}\end{smatrix} $}
      \State \label{step:zerocol} Let $k$ be a zero column of 
$ \begin{smatrix}  \mat{C}_{i,i-1} \\ \dots \\ \mat{C}_{t,i-1}\end{smatrix}$
      \State \label{step:movecol} Move col. $j$ in $\begin{smatrix}  \mat{C}_{i,i-2}\\ \vdots\\\mat{C}_{t,i-2}   \end{smatrix}$ to 
   col. $k$ in  $\begin{smatrix}  \mat{C}_{i,i-1}\\ \vdots\\\mat{C}_{t,i-1}   \end{smatrix}$.
      \State $\mat{T}\leftarrow (\mat{I}_n + \mat{\Delta}^{(k,j)}-\mat{\Delta}^{(k,k)})\times \mat{T}$
   \EndFor
\EndFor

\State$\mat{S}\leftarrow\mat{C}=\begin{smatrix}
  \mat{0}& \\
  \mat{C}_{21} & \mat{0} \\
       & \ddots& \ddots\\
       &       & \mat{C}_{t,t-1} & \mat{0}
\end{smatrix}$
\Return $(\mat{D},\mat{S}, \mat{T}, \mathcal{Q})$
\end{algorithmic}
}
\end{algorithm}
\begin{figure}[h]
  \centering
   \includegraphics[height=\myfigsize]{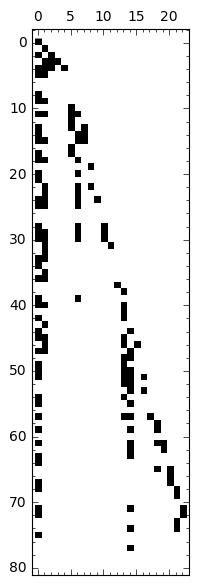}
   \includegraphics[height=\myfigsize]{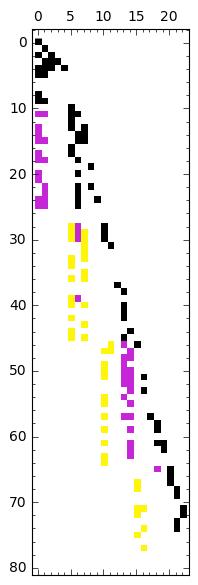}
     \includegraphics[width=\myfigsize]{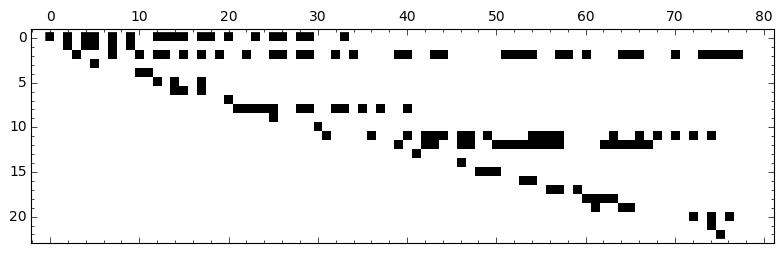}\\
     \includegraphics[width=\myfigsize]{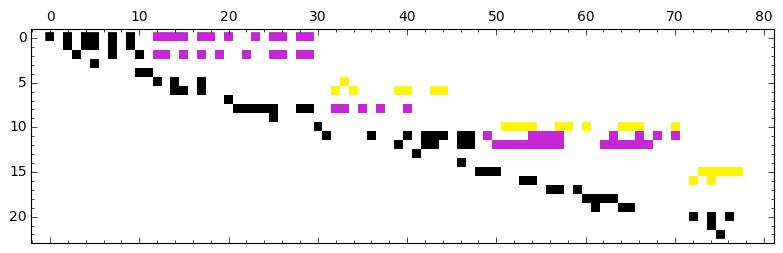}\\
  \caption{Support of the matrices $\mathcal{L}\mathcal{Q}$ (left),
    $\mathcal{P}\mathcal{U}$ (center) and of the corresponding
    compact \bruhat (right and bottom) for the matrix of Figure~\ref{fig:bruhatstorage}. In
  the compact \bruhat: $\mat{D}$ is in black and
  $\mat{S}$ in magenta and yellow; those rows and columns 
  moved at step~\ref{step:movecol} of Algorithm~\ref{alg:compactbruhat} are
   in yellow.}
  \label{fig:compactbruhat}
\end{figure}
%

\begin{lemma}
  Algorithm~\ref{alg:compactbruhat} computes a tuple $(\mat{D},\mat{S},\mat{T},\mathcal{Q})$
where $\mathcal{Q}$ is a permutation matrix putting $\mathcal{L}$ in column
echelon form, $\mat{T} \in \{0,1\}^{r\times r}$,  $\mat{D}=\text{Diag}(\mat{D}_1,\dots,\mat{D}_{t})$, $\mat{S}=
\begin{smatrix}
  \mat{0}\\
  \mat{S}_2 & \mat{0}\\
   & \ddots &\ddots\\
   & &\mat{S}_t
\end{smatrix}
$  where each $\mat{D}_i$ and $\mat{S}_i$ is $k_i\times s$ for $k_i\geq s$ and
$\sum_{i=1}^tk_i=n$. This tuple is the compact \bruhat for $\mathcal{L}$ and satisfies
$
\mathcal{L} =
\begin{bmatrix}
 \mat{D} + \mat{S} \mat{T} & \mat{0}_{n\times (n-r)}
\end{bmatrix}
\mathcal{Q}^T
$.
\end{lemma}
\begin{proof}
First, note that for every $i$, the dimensions of the blocks $\mat{S}_i$ and $\mat{D}_i$ are that of the
block $\mat{C}_{ii}$. This block contains $s$ pivots, hence $k_i\geq s$.
We then prove that there always exists a zero column to pick at
step~\ref{step:zerocol}.
The loci of the possible non-zero elements in $\mathcal{L}$ are column segments
below a pivot and above the anti-diagonal. From Lemma~\ref{lem:size}, these
segments have the property that each row of $\mathcal{L}$ is intersected by no
more than $s$ of them.
This property is preserved by column permutation, and still holds on the matrix  $\mat{C}$.
In the first row of $\begin{bmatrix} \mat{C}_{i1}&\dots&
  \mat{C}_{ii}\end{bmatrix}$, there is a pivot located in the block
$\mat{C}_{ii}$. Hence there is at most $s-1$ such segments intersecting $\begin{bmatrix} \mat{C}_{i1}&\dots&
  \mat{C}_{i,i-1}\end{bmatrix}$. These $s-1$ columns can all be gathered in the block
$\mat{C}_{i,i-1}$ of column dimension~$s$.

There only remains to show that $\mat{S}\mat{T}$ is the matrix $\mat{C}$ of
step~\ref{step:CminusD}. For every pair of indices $(j,k)$ selected in
loop~\ref{step:loop}, right multiplication by
$(\mat{I}_n+\mat{\Delta}^{(k,j)}-\mat{\Delta}^{(k,k)})$ adds up column $k$ to column
$j$ and zeroes out column $k$. On matrix $\mat{S}$, this has the effect of
reverting each operation done at step~\ref{step:movecol} in the reverse order
of  the loop~\ref{step:loop}.
\end{proof}


A compact representation of $\mathcal{U}$ is obtained in Lemma~\ref{lem:compactU} by running
Algorithm~\ref{alg:compactbruhat} on  $\mathcal{U}^T$ and transposing its
output.

\begin{lemma}\label{lem:compactU}
  There exist a tuple $(\mat{D},\mat{S},\mat{T},\mathcal{P})$ called the compact
  \bruhat for $\mathcal{U}$ such that
 $\mathcal{P}$ is a permutation matrix putting $\mathcal{U}$ in row echelon
form, $\mat{T} \in \{0,1\}^{r\times r}$,  $\mat{D}=\text{Diag}(\mat{D}_1,\dots,\mat{D}_{t})$, $\mat{S}=
\begin{smatrix}
  \mat{0}&  \mat{S}_2 \\
   & \ddots &\ddots\\
   &&\mat{0}&\mat{S}_t
\end{smatrix}
$  where each $\mat{D}_i$ and $\mat{S}_i$ is $s\times k_i$ for $k_i\geq s$ and
$\sum_{i=1}^tk_i=n$ and 
$
\mathcal{U} =
\mathcal{P}^T \begin{bmatrix}
 \mat{D} +  \mat{T}\mat{S} \\ \mat{0}_{(n-r)\times n}
\end{bmatrix}
$.
\end{lemma}

According to \eqref{eq:ALEU}, the reconstruction of the initial matrix~$\mat{A}$, from the compact
Bruhat generators, writes
\begin{equation}\label{eq:Acompact}
\mat{A}=(\mat{D}_\mathcal{L}+\mat{S}_\mathcal{L}\mat{T}_\mathcal{L}) \mat{R} (\mat{D}_\mathcal{U}+\mat{T}_\mathcal{U}\mat{S}_\mathcal{U})
\end{equation}
where $\mat{R}$ is the leading $r\times r$ sub-matrix of
$\mathcal{Q}^T\mathcal{E}^T\mathcal{P}^T$. As it has full rank, it is a permutation matrix.

This factorization is a compact version of the generalized Bruhat
decomposition~\cite{MH07,DPS15:JSC}: the left factor is a column echelon form,
the right factor a row echelon form.
\section{Cost of computing with the new generators}\label{sec:cost}

\subsection{Computation of the generators}
\subsubsection{The binary tree generators}
Let $T_1(n,s)$ denote the cost of the computation of the  binary tree generator
for an $n\times n$ matrix of order of quasiseparability $s$.
It satisfies the recurrence relation $T_1(n,s) = K_\omega s^{\omega-2}\left(\frac{n}{2}\right)^2 +
2T_1(n/2,s)$, which solves in 
$$
T(n,s) = \frac{K_\omega}{2} s^{\omega-2}n^2 \text{ with }
K_\omega = \frac{2^{\omega-2}}{(2^{\omega}-2)(2^{\omega-2}-1)}\mat{C}_\omega
$$
where $C_\omega$ is the leading constant of the complexity of matrix
multiplication~\cite{DPS13}.

\subsubsection{The \bruhat}

We propose in Algorithm~\ref{alg:bruhatgenerator} an evolution of Algorithm~\ref{alg:LTElim}
to compute the factors of the \bruhat.
\begin{algorithm}[htb]
{\small
  \caption{LT-Bruhat}
\label{alg:bruhatgenerator}
  \begin{algorithmic}[1]
    \Require{$\mat{A}$: an $n\times n$ matrix}
    \Ensure{$(\mathcal{L},\mathcal{E},\mathcal{U})$: a \bruhat for the left
      triangular part of $\mat{A}$}
    \State \textbf{if} $n=1$ \textbf{then} \textbf{return} $([0],[0],[0])$
    \State {Split $\mat{A} = \begin{bmatrix} \mat{A_{1}} & \mat{A_{2}}\\\mat{A_{3}} \end{bmatrix}$ 
      where $\mat{A_{3}}$ is $\lfloor \frac{n}{2} \rfloor \times \lfloor
      \frac{n}{2}\rfloor$}
    \State Decompose $\mat{A_{1}} =
    \mat{P_1} \begin{bmatrix}\mat{L_1}\\\mat{M_1}\end{bmatrix}\begin{bmatrix}\mat{U_1}&\mat{V_1}\end{bmatrix} \mat{Q_1}$ \Comment{$\pluq(\mat{A_{1}})$}
    \State $\mat{R_1}  \leftarrow \mat{P_1} \begin{bmatrix}
      \mat{I_{r_1}}\\&\mat{0}\end{bmatrix}\mat{Q_1}$ where $r_1=\text{rank}(\mat{A_1})$.
    \State $\begin{bmatrix} \mat{B_1}\\ \mat{B_2}\end{bmatrix}
    \leftarrow \mat{P_1}^T\mat{A_{2}}$ \Comment{$\texttt{PermR}(\mat{A_2},\mat{P_1}^T)$}
    \State $
    \begin{bmatrix}
      \mat{C_1}&\mat{C_2}
    \end{bmatrix}
    \leftarrow \mat{A_{3}}\mat{Q_1}^T$ \Comment{$\texttt{PermC}(\mat{A_3},\mat{Q_1}^T)$}
    \State Here $A =
    \left[\begin{array}{cc|c}
        \mat{L_1} \backslash \mat{U_1}& \mat{V_1}& \mat{B_1}\\
        \mat{M_1}               & \mat{0}  & \mat{B_2}\\
        \hline
        \mat{C_1}               & \mat{C_2}& \\
      \end{array}\right]$. \label{step:partialPLUQ}
    \State $\mat{D}\leftarrow \mat{L_1}^{-1}\mat{B_1}$ \Comment{$\trsm(\mat{L_1},\mat{B_1})$}
    \State $\mat{E}\leftarrow \mat{C_1}\mat{U_1}^{-1}$ \Comment{$\trsm(\mat{C_1},\mat{U_1})$}
    \State $\mat{F}\leftarrow \mat{B_2}-\mat{M_1}\mat{D}$ \Comment{$\MM(\mat{B_2},\mat{M_1},\mat{D})$}
    \State $\mat{G}\leftarrow \mat{C_2}-\mat{E}\mat{V_1}$ \Comment{$\MM(\mat{C_2},\mat{E},\mat{V_1})$}
    \State Here $\mat{A}=
    \left[\begin{array}{cc|c}
        \mat{L_1} \backslash \mat{U_1}& \mat{V_1}& \mat{D}\\
        \mat{M_1}               & \mat{0}  & \mat{F}\\
        \hline
        \mat{E}               & \mat{G}& \\
      \end{array}\right]$.
    
    \State $\mat{H} \leftarrow \mat{P_1}  \begin{bmatrix} \mat{0}_{r_1 \times \frac{n}{2}} \\ \mat{F} \end{bmatrix}$
    \State $\mat{I} \leftarrow  \begin{bmatrix} \mat{0}_{r_1 \times \frac{n}{2}} & \mat{G} \end{bmatrix} \mat{Q_1}$
    \State\label{step:firstLT} $(\mathcal{L}_2,\mathcal{E}_2,\mathcal{U}_2) \leftarrow
    \texttt{LT-Bruhat}(\mat{H})$ 
    \State\label{step:secondLT} $(\mathcal{L}_3,\mathcal{E}_3,\mathcal{U}_3) \leftarrow
    \texttt{LT-Bruhat}(\mat{I})$ 
    \State $\mathcal{L} \leftarrow
\LTP \left(    \begin{bmatrix}      \mat{P}_1\\&\mat{I}_{\frac{n}{2}}    \end{bmatrix} 
    \begin{bmatrix}
      \mat{L}_1 \\
      \mat{M}_1 &\mat{0}\\
      \mat{E} & \mat{0}
    \end{bmatrix}
    \begin{bmatrix}
      \mat{Q}_1 \\&\mat{I}_{\frac{n}{2}}
    \end{bmatrix}\right)
+
\begin{bmatrix}  \mat{0}&\mathcal{L}_2\\\mathcal{L}_3\end{bmatrix}
$
    \State $\mathcal{U} \leftarrow
    \begin{bmatrix}
      \mat{P}_1
      \begin{bmatrix}\mat{U}_1&V_1\\\mat{0}&\mat{0} \end{bmatrix} \mat{Q}_1 & \LTP(\mat{P}_1
      \begin{bmatrix} \mat{D}\\\mat{0}  \end{bmatrix}) \\
      \mat{0}&\mat{0}
    \end{bmatrix}
+
\begin{bmatrix}  \mat{0}&\mathcal{U}_2\\\mathcal{U}_3\end{bmatrix}
$\label{step:calU}
    \State $\mathcal{E} \leftarrow  \begin{bmatrix}     \mathcal{E}_1 & \mathcal{E}_2\\ \mathcal{E}_3   \end{bmatrix}$
   \State \textbf{return} $(\mathcal{L},\mathcal{E},\mathcal{U})$
  \end{algorithmic}
}
\end{algorithm}

\begin{theorem}
  For any $n\times n$ matrix $\mat{A}$ with a left triangular part of
  quasiseparable order $s$, Algorithm~\ref{alg:bruhatgenerator} computes the \bruhat of the
  left triangular part of $\mat{A}$ in $\GO{s^{\omega-2}n^2}$ field operations.
\end{theorem}
\begin{proof}
The correctness of $\mathcal{E}$ is proven in Theorem~\ref{th:LTRPM}.
We will prove by induction the correctness of $\mathcal{U}$, noting that the correctness of
$\mathcal{L}$ works similarly.

Let $\mat{H}=\mat{P}_2\mat{L}_2\mat{U}_2\mat{Q}_2$ and
$\mat{I}=\mat{P}_3\mat{L}_3\mat{U}_3\mat{Q}_3$ be PLUQ decompositions of
$\mat{H}$ and $\mat{I}$ revealing their rank profile matrices. 
Assume that Algorithm LT-Bruhat is correct in the two recursive calls~\ref{step:firstLT}
and~\ref{step:secondLT}, that is 
$$
\begin{array}{ll}
\mathcal{U}_2 = \LTP(\mat{P}_2\begin{bmatrix} \mat{U}_2\\\mat{0}\end{bmatrix} \mat{Q}_2), & 
\mathcal{U}_3 = \LTP(\mat{P}_3\begin{bmatrix} \mat{U}_3\\\mat{0}\end{bmatrix} \mat{Q}_3),\\
\mathcal{L}_2 = \LTP(\mat{P}_2\begin{bmatrix} \mat{L}_2&\mat{0}\end{bmatrix} \mat{Q}_2), &
\mathcal{L}_3 = \LTP(\mat{P}_3\begin{bmatrix} \mat{L}_3&\mat{0}\end{bmatrix} \mat{Q}_3).\\
\end{array}
$$

At step~\ref{step:partialPLUQ}, we have
{\small
\begin{eqnarray*}
\begin{bmatrix}
  \mat{A}_1 & \mat{A}_2\\
  \mat{A}_3 &*
\end{bmatrix}
 &=&
\begin{bmatrix}\mat{P}_1\\&\mat{I}_{\frac{n}{2}}\end{bmatrix}
\left[\begin{array}{cc|c}
  \mat{L}_1&&\\ 
  \mat{M}_1 &\mat{I}_{\frac{n}{2}-r_1}\\
\hline
  \mat{E} & \mat{0}&\mat{I}_{\frac{n}{2}}
\end{array}\right] \times \\
&&
\left[\begin{array}{cc|c}
  \mat{U}_1 & \mat{V}_1 &  \mat{D}\\
&\mat{0} &  \mat{F}\\
\hline
&\mat{G}&
\end{array}
\right]
\begin{bmatrix} \mat{Q}_1\\ & \mat{I}_{\frac{n}{2}} \end{bmatrix}
\end{eqnarray*}
}
As the first $r_1$ rows of $\mat{P}_1^T\mat{H}$ are zeros, there exists $\mat{\bar P}_2$ a
permutation matrix and $\mat{\bar L}_2$, a lower triangular matrix, such that
$\mat{P}_1^T\mat{P}_2\mat{L}_2 = 
\begin{bmatrix} \mat{0}_{r_1\times \frac{n}{2}} \\ \mat{\bar P}_2 \mat{\bar L}_2 \end{bmatrix}$. 
Similarly, there exsist $\mat{\bar Q}_3$, a
permutation matrix and $\mat{\bar U}_3$, an upper triangular matrix, such that $\mat{U}_3\mat{Q}_3\mat{Q}_1^T =
\begin{bmatrix}
  \mat{0}_{\frac{n}{2}\times r_1} &
  \mat{\bar U}_3 \mat{\bar Q}_3
\end{bmatrix}$.
Hence
{\small
\begin{eqnarray*}
\begin{bmatrix}
  \mat{A}_1 & \mat{A}_2\\
  \mat{A}_3 &*
\end{bmatrix}
 &=&
\begin{bmatrix}\mat{P}_1
\\&\mat{P}_3\end{bmatrix}
\left[\begin{array}{cc|c}
  \mat{L}_1&&\\ 
  \mat{M}_{1} & \mat{\bar P}_2\mat{\bar L}_2\\
\hline
  \mat{P}_3^T\mat{E} & \mat{0}&\mat{L}_3
\end{array}\right] 
\times\\
&&\left[\begin{array}{cc|c}
  \mat{U}_1 & \mat{V}_1 &  \mat{D}\mat{Q}_2^T\\
&\mat{0} &  \mat{U_2}\\
&\mat{\bar U}_3\mat{\bar Q}_3
\end{array}
\right]
\begin{bmatrix} \mat{Q}_1\\ & \mat{Q}_2 \end{bmatrix}\\
\end{eqnarray*}
}
Setting $\mat{N}_1 = \mat{\bar P}_2^T\mat{M}_1$ and $\mat{W}_1 = \mat{V}_1
\mat{\bar Q}_3^T$, we have
{\small
\begin{eqnarray*}
\begin{bmatrix}
  \mat{A}_1 & \mat{A}_2\\
  \mat{A}_3 &*
\end{bmatrix}
&=&
\begin{bmatrix}\mat{P}_1
  \begin{bmatrix}
    \mat{I}_{r_1}\\&\mat{\bar P}_2
  \end{bmatrix}
\\&\mat{P}_3\end{bmatrix}
\left[\begin{array}{ccc}
  \mat{L}_1\\ 
  \mat{N}_{1} & \mat{\bar L}_2\\
\hline
  \mat{E} &\mat{0}& \mat{\mat{L_3}}
\end{array}\right] 
\times\\
&&\left[\begin{array}{cc|c}
  \mat{U}_1 & \mat{W}_1 &  \mat{D}\mat{Q}_2^T\\
&\mat{0} &  \mat{U_2}\\
&\mat{\bar U}_3
\end{array}
\right]
\begin{bmatrix}
  \begin{bmatrix}
    \mat{I}_{r_1}\\&\mat{\bar Q}_3
  \end{bmatrix}
\mat{Q}_1\\ & \mat{Q}_2 \end{bmatrix}.
\end{eqnarray*}
}
A PLUQ of $
\begin{smatrix}
  \mat{A}_1 & \mat{A}_2\\ \mat{A}_3
\end{smatrix}
$ revealing its rank profile matrix is then obtained from this decomposition by
a row block cylic-shift on the second factor and a column block cyclic shift on
the third factor as in \cite[Algorithm~1]{DPS13}.

Finally,
{\small
\begin{equation*}
\begin{split}
\mat{P} \begin{bmatrix}  \mat{U}\\0\end{bmatrix}\mat{Q}
=
\begin{bmatrix}  \mat{P}_1\\&\mat{I}_{\frac{n}{2}}\end{bmatrix}
\begin{bmatrix}
  \mat{U}_1 & \mat{V}_1 & \mat{D} \\
  & \mat{0} & \mat{\bar P}_2 \mat{U}_2\mat{Q}_2 \\
  &\mat{P}_3 \mat{\bar U}_3  \mat{\bar Q}_3\\
\mat{0}& \mat{0}&\mat{0}
\end{bmatrix}
\begin{bmatrix}
  \mat{Q}_1\\&\mat{I}_{\frac{n}{2}}
\end{bmatrix}\\
=
\begin{bmatrix}
\mat{P}_1  \begin{bmatrix}
      \mat{U}_1 & \mat{V}_1\\
      \mat{0}&\mat{0}
  \end{bmatrix}
\mat{Q}_1 & \mat{P}_1
\begin{bmatrix}
  \mat{D}\\\mat{0}
\end{bmatrix}
 \\
  \mat{0}& \mat{0} \
\end{bmatrix}
+\begin{bmatrix}
& \mat{P}_2
\begin{bmatrix}
  \mat{U}_2\\ \mat{0}
\end{bmatrix}
 \mat{Q}_2\\
\mat{P}_3
\begin{bmatrix}
  \mat{U}_3\\\mat{0}
\end{bmatrix}
\mat{Q}_3
\end{bmatrix}.
\end{split}
\end{equation*}
}
Hence
{\small
\begin{equation*}
\LTP(\mat{P}\mat{U}\mat{Q}) =
\begin{bmatrix}
 \mat{P}_1
 \begin{bmatrix}
   \mat{U}_1 & \mat{V}_1 \\\mat{0}&\mat{0}
 \end{bmatrix}\mat{Q}_1
 & \LTP(\mat{P}_1
 \begin{bmatrix} \mat{D}\\\mat{0} \end{bmatrix}) \\
   \mat{0} & \mat{0} \\
\end{bmatrix} 
+ \begin{bmatrix}
&  \mathcal{U}_2\\
\mathcal{U}_3
\end{bmatrix}.
\end{equation*}
}


  
The complexity analysis is exactly that of Theorem~\ref{th:LTRPM}.
\end{proof}

The computation of a compact \bruhat is obtained by combining
Algorithm~\ref{alg:bruhatgenerator} with Algorithm~\ref{alg:compactbruhat}.
\subsection{Applying a vector}
For the three generators proposed earlier, the application of a vector to the
corresponding left triangular matrix takes the same amount of field
operations as the number of coefficients used for its representation. This yields a
cost of $\GO{n(r_L \log\frac{n}{r_L} + r_U \log \frac{n}{r_U})}$ field
operations for multiplying a vector to an $(r_L,r_U)$-quasiseparable matrix
using the binary tree PLUQ generator and $\GO{n(r_L+r_U)}$ using either one of
the \bruhat or its compact variant.

\subsection{Multiplying two left-triangular matrices}\label{sec:mulLT}

\subsubsection{The binary tree PLUQ generator}

Let $T_\text{RL}(n,s)$ denote the cost of multiplying a dense $s\times n$ matrix
by a left triangular quasiseparable matrix of order $s$. The natural divide and
conquer algorithm yields the recurrence formula:
\begin{equation*}
T_{\text{RL}} (n,s) =  2T_{\text{RL}}(n/2,s) + \GO{ns^{\omega-1}} = \GO{ns^{\omega-1}\log{\frac{n}{s}}}.
\end{equation*}
Let $T_\text{PL}(n,s)$ denote the cost of multiplying a PLUQ decomposition of
dimension n and rank $s\leq n/2$ with a  left triangular quasiseparable matrix
of order $s$. The product can be done in 
\begin{eqnarray*}
T_{\text{PL}} (n,s) &=& T_{\text{RL}} (n,s)+ \GO{n^2s^{\omega-2}}= \GO{n^2s^{\omega-2}}.
\end{eqnarray*}
Lastly, let $T_\text{LL}(n,s)$ denote the cost of multiplying two left-triangular matrices of quasiseparability
order $s$. Again the natural recursive algorithm yields:
\begin{eqnarray*}
  T_{\text{LL}} (n,s) &=&  2T_{\text{LL}} (n/2,s) + 2T_{\text{PL}} (n/2,s) + \GO{n^2s^{\omega-2}}\\
  &=&\GO{n^2s^{\omega-2}}
\end{eqnarray*}

\subsubsection{The \bruhat}

Using the decomposition~\eqref{eq:Acompact}, the product of two left triangular
matrices writes
$  \mat{A}\times \mat{B} =\mat{C}_\mat{A} \mat{R}_\mat{A} \mat{E}_\mat{A} \times \mat{C}_\mat{B} \mat{R}_\mat{B} \mat{E}_\mat{B} 
 $ where $\mat{C}_\mat{X} = \mat{D}_{\mathcal{L}_\mat{X}}+\mat{S}_{\mathcal{L}_\mat{X}}\mat{T}_{\mathcal{L}_\mat{X}}$
and
$\mat{E}_\mat{X} =
\mat{D}_{\mathcal{U}_\mat{X}}+\mat{T}_{\mathcal{U}_\mat{X}}\mat{S}_{\mathcal{U}_\mat{X}}$  for $\mat{X}\in\{\mat{A},\mat{B}\}$.
We will compute it using the following parenthesizing:
\begin{equation}
  \label{eq:parenthesizing}
    \mat{A}\times \mat{B}
    =\mat{C}_\mat{A} (\mat{R}_\mat{A} (\mat{E}_\mat{A} \times \mat{C}_\mat{B}) \mat{R}_\mat{B}) \mat{E}_\mat{B}. 
\end{equation}

The product $\mat{E}_\mat{A} \times \mat{C}_\mat{B}=(\mat{D}_{\mathcal{U}_\mat{A}}+\mat{T}_{\mathcal{U}_\mat{A}}\mat{S}_{\mathcal{U}_\mat{A}})( \mat{D}_{\mathcal{L}_\mat{B}}+\mat{S}_{\mathcal{L}_\mat{B}}\mat{T}_{\mathcal{L}_\mat{B}})$
only consists in multiplying together block diagonal or sub-diagonal matrices
$n\times r_B$ or $r_A\times n$. We will describe the product of two block
diagonal matrices (flat times tall); the other cases with sub-diagonal matrices work similarly.

Each term to be multiplied is decomposed in a grid of $s\times s$ tiles (except
at the last row and column positions). In this
grid, the non-zero blocks are non longer in a block-diagonal layout: in a flat matrix, the leading block of a
block row may lie at the same block column position as the trailing block of its
preceding block row,
as shown in Figure~\ref{fig:blocks}. 
\begin{figure}[h]
  \centering
  \includegraphics[width=.52\columnwidth]{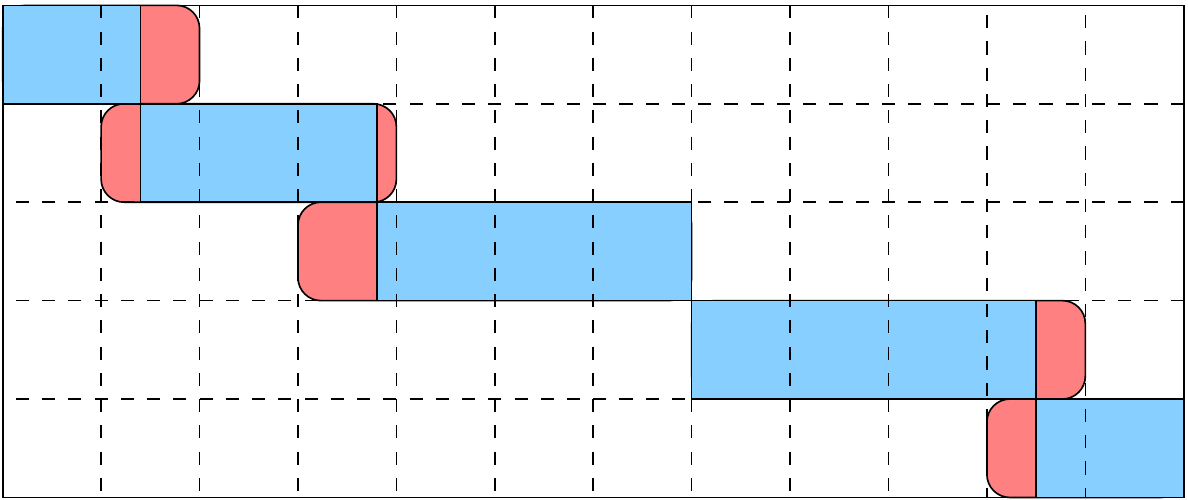}
  \caption{Aligning a block diagonal matrix (blue) on an $s\times s$ grid. Each
    block row of the aligned structure (red) may overlap with the previous and next
    block rows on at most one $s\times s$ tile on each side.
}
  \label{fig:blocks}
\end{figure}
However, since $k_i\geq s$ for all $i$, no more than two
consecutive block rows of a flat matrix lie in the same block column.
Consequently these terms can be decomposed as a sum of two block diagonal
matrices aligned on an $s\times s$ grid. Multiplying two such matrices costs
$\GO{s^{\omega-1}n}$ which is consequently also the cost of computing  the product $\mat{E}_\mat{A}\mat{C}_\mat{B}$.
After left and right multiplication by the permutations $\mat{R}_\mat{A}$ and
$\mat{R}_\mat{B}$, this $r_\mat{A}\times r_\mat{B}$
dense matrix is multiplied to the left by $\mat{C}_\mat{A}$. This costs
$\GO{nr_{B}s^{\omega-2}}$. Lastly, the right multiplication by
$\mat{E}_{B}$ of the resulting $n\times r_\mat{A}$ matrix costs
$\GO{n^2s^{\omega-2}}$ which dominates the overall cost.

\subsection{Multiplying two quasiseparable matrices}
Decomposing each multiplicand into its upper, lower and diagonal terms, a
product of two quasiseparable matrices writes
$
\mat{A}\times \mat{B} = (\mat{L}_\mat{A} +\mat{D}_\mat{A} +\mat{U}_\mat{A}) (\mat{L}_\mat{B} +\mat{D}_\mat{B} +\mat{U}_\mat{B}).
$
Beside the scaling by diagonal matrices, all other operations involve a
product between any combination of lower an upper triangular matrices, which in turn
translates into products of left triangular matrices and $\mat{J}_n$ as shows in Table~\ref{tab:loweupper}.
\begin{table}[h]
{\small
  \centering
\begin{tabular}[h]{c|cc}
\toprule
$\times$  & Lower & Upper\\
\midrule
Lower & $\mat{J}_n \times\text{Left} \times\mat{J}_n\times \text{Left}$ & $\mat{J}_n\times \text{Left}\times \text{Left}\times\mat{J}_n$\\
Upper & $\text{Left} \times \mat{J}_n\times \mat{J}_n \times\text{Left}$ & $\text{Left}\times\mat{J}_n\times \text{Left}\times\mat{J}_n$\\
\bottomrule
\end{tabular}
  \caption{Reducing products of lower and upper to products of left triangular matrices.}\label{tab:loweupper}
}
\end{table}
The complexity of section~\ref{sec:mulLT} directly applies for the computation of $\text{Upper}\times \text{Lower}$ and
$\text{Lower}\times\text{Upper}$ products. For the other products, a
$\mat{J}_n$ factor has to be added 
between the $\mat{E}_\mat{A}$ and $\mat{C}_\mat{B}$ factors in the innermost
product of~\eqref{eq:parenthesizing}. As reverting the row order of
$\mat{C}_\mat{B}$ does not impact the cost of computing this product,
the same complexity applies here too.

\begin{theorem}
  Mutliplying two quasiseparable matrices of order respectively
  $(l_\mat{A},u_\mat{A})$ and  $(l_\mat{B},u_\mat{B})$ costs $\GO{n^2s^{\omega-2}}$ field operations
  where $s=\max(l_\mat{A},u_\mat{A},l_\mat{B},u_\mat{B})$,  using either one of the binary tree or the compact \bruhat.
\end{theorem}

\section{Perspectives}

The algorithms proposed for multiplying two quasiseparable matrices output
a dense $n\times n$ matrix  in time $\GO{n^2s^{\omega-2}}$ for $s=\max(l_\mat{A},u_\mat{A},l_\mat{B},u_\mat{B})$.
However, the product is also a quasiseparable matrix, of order
$({l_\mat{A}}+{l_\mat{B}},{u_\mat{A}}+{u_\mat{B}})$~\cite[Theorem~4.1]{EiGo99},
which can be represented by a \bruhat with only
$\GO{n({l_\mat{A}}+{l_\mat{B}}+{u_\mat{A}}+{u_\mat{B}})}$ coefficients. 
A first natural question is thus to find an algorithm computing this
representation from the generators of $\mat{A}$
and $\mat{B}$ in time $\GO{ns^{\omega-1}}$.

Second, a probabilistic algorithm~\cite[\S~7]{DPS15:JSC} reduces the complexity
of computing the rank profile matrix to $\SO{n^2+r^\omega}$. It is not clear whether it can be applied to
compute a compact Bruhat generator in time $\SO{n^2 + \max(l_\mat{A},u_\mat{A})^\omega}$.

\section*{Note (added Sept. 16, 2016.)}
Equation~\eqref{eq:parenthesizing} for the multiplication of two Bruhat
generators is missing the $\LTP{}$ operators, and is therefore incorrect. The target
complexities can still be obtained by slight modification of the algorithm:
computing the inner-most product $\mat{E}_\mat{A}\times \mat{C}_\mat{B}$ as an
unevaluated sum of blocks products. This will be detailed in a follow-up paper.

\section*{Acknowledgment}

We thank Paola Boito for introducing us to the field of 
quasiseparable matrices and  two anonymous referees for pointing us to the HSS and the Givens
weight representations.
We acknowledge the financial support from the \href{//http://hpac.gforge.inria.fr/}{HPAC project}
(ANR~11~BS02~013) and from the \href{http://opendreamkit.org/}{OpenDreamKit} \href{https://ec.europa.eu/programmes/horizon2020/}{Horizon 2020} \href{https://ec.europa.eu/programmes/horizon2020/en/h2020-section/european-research-infrastructures-including-e-infrastructures}{European Research Infrastructures} project (\href{http://cordis.europa.eu/project/rcn/198334_en.html}{\#676541}).

\bibliography{quasisep}
\bibliographystyle{plain}

\end{document}